\documentclass[fleqn]{llncs}

\usepackage{amsmath}
\usepackage{amssymb}

\usepackage[utf8]{inputenc}
\usepackage[T1]{fontenc} 

\usepackage{wrapfig}

\usepackage{todonotes}

\usepackage{tikz}
\usetikzlibrary{automata}
\usetikzlibrary{backgrounds}

\newcommand{\nats}{\mathbb{N}}
\newcommand{\size}[1]{|#1|}

\renewcommand{\epsilon}{\varepsilon}
\renewcommand{\phi}{\varphi}

\newcommand{\myquot}[1]{``#1''}

\newcommand{\proj}[1]{\pi_1(#1)}
\newcommand{\set}[1]{\{#1\}}

\newcommand{\aut}{\mathcal{A}}
\newcommand{\autproduct}{\mathcal{P}}
\newcommand{\auttrack}{\mathcal{T}}
\newcommand{\autpow}{\mathrm{pow}}

\newcommand{\ops}{\mathrm{Ops}}
\newcommand{\inc}[1]{#1 := #1 + 1}
\newcommand{\reset}[1]{#1 := 0}
\newcommand{\maxx}[3]{#1 := \max(#2,#3)}
\newcommand{\val}{\nu}

\newcommand{\lbls}{\Lambda}

\newcommand{\flatten}{\mathrm{flat}}
\newcommand{\run}{\rho}

\newcommand{\delaygame}[1]{\Gamma\!_{f}(#1)}
\newcommand{\SigmaO}{\Sigma_O}
\newcommand{\SigmaI}{\Sigma_I}
\newcommand{\stratO}{\tau_O}
\newcommand{\stratI}{\tau_I}
\newcommand{\p}{p}

\newcommand{\exptime}{\textsc{Exptime}}

\newcommand{\resolve}{r}
\newcommand{\game}{\mathcal{G}}
\newcommand{\wit}[1]{W_{#1}}
\newcommand{\witmap}[1]{\resolve_{#1}}
\newcommand{\init}{v_I}
\newcommand{\curlyR}{\mathfrak{R}}
\newcommand{\dom}{\mathrm{dom}}

\newcommand{\pow}{{2^{Q_\autproduct}}}
\newcommand{\rep}{\mathrm{rep}}
\newcommand{\win}{\mathrm{Win}}

\newcommand{\eqop}{\equiv_{\mathrm{ops}}}
\newcommand{\eqword}{\equiv_\aut}

\newcommand{\eqclass}[1]{[#1]}
\newcommand{\eqclassword}[1]{[#1]_{\eqword}}

\newcommand{\quotient}{\hspace{-0pt}\slash\hspace{-2pt}\eqword}

\newcommand{\U}{\mathrm{U}}

\title{Delay Games with\\WMSO$+$U Winning Conditions\thanks{A short version appears in the proceedings of CSR 2015~\cite{Zimmermann15}.}}

\author{Martin Zimmermann\thanks{Supported by the DFG projects \myquot{TriCS} (ZI 1516/1-1) and \myquot{AVACS} (SFB/TR 14).}}

\institute{Reactive Systems Group, Saarland University, Germany\\
 \email{zimmermann@react.uni-saarland.de}}

\begin{document}

\maketitle

\begin{abstract}
\noindent Delay games are two-player games of infinite duration in which one player may delay her moves to obtain a lookahead on her opponent's moves. We consider delay games with winning conditions expressed in weak monadic second order logic with the unbounding quantifier, which is able to express (un)boundedness properties.

We show that it is decidable whether the delaying player has a winning strategy using bounded lookahead and give a doubly-exponential upper bound on the necessary lookahead. In contrast, we show that bounded lookahead is not always sufficient to win such a game.
\end{abstract}


\section{Introduction}
\label{intro}
Many of today's problems in computer science are no longer concerned with
programs that transform data and then terminate, but with non-terminating
reactive systems which have to interact with a possibly antagonistic
environment for an unbounded amount of time. The framework of infinite
two-player games is a powerful and flexible tool to verify and synthesize such
systems. The seminal theorem of B\"uchi and Landweber~\cite{BuechiLandweber69}
states that the winner of an infinite game on a finite arena with an $\omega$-regular
winning condition can be determined and a corresponding finite-state winning
strategy can be constructed effectively.

Ever since, this result was extended along different dimensions, e.g., the
number of players, the type of arena, the type of winning condition, the type
of interaction between the players (alternation or concurrency), zero-sum or
non-zero-sum, and complete or incomplete information. In this work, we consider
two of these dimensions, namely more expressive winning conditions and the
possibility for one player to delay her moves.

\paragraph{Delay Games.} In a delay game, one of the players can postpone her moves for some time,
thereby obtaining a lookahead on her opponent's moves. This allows her to
win some games which she loses without lookahead, e.g., if her first move
depends on the third move of her opponent. Nevertheless, there are
winning conditions that cannot be won with any finite lookahead, e.g., if her
first move depends on every move of her opponent. Delay
arises naturally when transmission of data in networks or components equipped
with buffers are modeled.

From a more theoretical point of view,
uniformization of relations by continuous functions~\cite{Thomas11,DBLP:conf/rex/ThomasL93,trakhtenbrot1973finite} can be expressed and
analyzed using delay games. We consider games in which two players pick letters from alphabets~$\SigmaI$
and $\SigmaO$, respectively, thereby producing two infinite sequences~$\alpha \in \SigmaI^\omega$
and $\beta\in \SigmaO^\omega$. Thus, a strategy for the second player induces a mapping
$\tau\colon \SigmaI^\omega\rightarrow\SigmaO^\omega$. It is winning for the
second player if $(\alpha,\tau(\alpha))$ is contained in the winning
condition~$L\subseteq \SigmaI^\omega\times \SigmaO^\omega$ for every $\alpha$. If  $\set{(\alpha,\tau(\alpha)) \mid \alpha \in \SigmaI^\omega} \subseteq L$, then $\tau$ uniformizes~$L$. 

In the classical setting of infinite games, in
which the players pick letters in alternation, the $n$-th letter of
$\tau(\alpha)$ depends only on the first $n$ letters of $\alpha$, i.e., $\tau$ satisfies a very strong notion of continuity. A strategy
with bounded lookahead, i.e., only finitely many
moves are postponed, induces a Lipschitz-continuous function $\tau$ (in the
Cantor topology on $\Sigma^\omega$) and a strategy with arbitrary lookahead
induces a continuous function (or equivalently, a uniformly continuous
function, as $\Sigma^\omega$ is compact).

Hosch and Landweber proved that it is decidable whether a game with $\omega$-regular
winning condition can be won with bounded lookahead~\cite{HoschLandweber72}. This result was improved by
Holtmann, Kaiser, and Thomas who showed that if a
player wins a game with arbitrary lookahead, then she wins already with
doubly-exponential bounded lookahead, and gave a streamlined decidability
proof yielding an algorithm with doubly-exponential running time~\cite{HoltmannKaiserThomas12}. Again, these
results were improved by giving an exponential upper bound on the necessary
lookahead and showing $\exptime$-completeness of the solution problem~\cite{KleinZimmermann15}. Going beyond $\omega$-regular winning
conditions by considering context-free conditions leads to undecidability and
non-elementary lower bounds on the necessary lookahead, even for very weak
fragments~\cite{FridmanLoedingZimmermann11}.

Thus, stated in terms of uniformization, Hosch and Landweber proved decidability of the
uniformization problem for $\omega$-regular relations by Lipschitz-continuous functions
and Holtmann et al.\ proved the equivalence of the existence of a
continuous uniformization function and the existence of a Lipschitz-continuous
uniformization function for $\omega$-regular relations. Furthermore, uniformization of context-free relations is undecidable, even with respect to Lip\-schitz-continuous functions.

In another line of work, Carayol and L\"oding considered the case of finite words~\cite{CarayolLoeding12}, and L\"oding and Winter~\cite{LoedingWinter14} considered the case of finite trees, which are both decidable. However, the nonexistence of MSO-definable choice functions on the infinite
binary tree~\cite{CarayolLoeding07,GS83} implies that uniformization fails for such trees.

\paragraph{WMSO$+$U.} In this work, we consider another class of conditions that go beyond the
$\omega$-regular ones. Recall that the $\omega$-regular languages are exactly
those that are definable in monadic second order logic (MSO)~\cite{Buechi62}. Recently, Boja\'{n}czyk has started a
program investigating the logic MSO$+$U, MSO extended with the unbounding
quantifier~$\U$. A formula~$\U X \phi(X)$ is satisfied, if there are
arbitrarily large \emph{finite} sets~$X$ such that $\phi(X)$ holds. MSO$+$U is
able to express all $\omega$-regular languages as well as non-regular languages like 
\[L = \set{ a^{n_0} b a^{n_1} b
a^{n_2} b \cdots \mid \limsup\nolimits_i n_i = \infty } \enspace .\]

Decidability of MSO$+$U turns out to be a delicate issue: there is no algorithm
that decides MSO$+$U on infinite trees and has a correctness proof using the
axioms of ZFC~\cite{BojanczykGMS14}. At the time of writing, an unconditional undecidability result for MSO$+$U on infinite words is presented~\cite{BPT15}. 

Even before these undecidability results were shown, much attention was being paid to fragments of the logic obtained by restricting the power of the second-order quantifiers. In particular, considering
weak\footnote{Here, the second-order quantifiers are restricted to finite
sets.} MSO with the unbounding quantifier (denoted by prepending a W) turned
out to be promising: WMSO$+$U on infinite words~\cite{Bojanczyk11} and on
infinite trees~\cite{BojanczykTorunczyk12} and WMSO$+$U with the path
quantifier (WMSO$+$UP) on infinite trees~\cite{Bojanczyk14} have equivalent automata models with decidable emptiness. Hence, these
logics are decidable.

For WMSO$+$U on infinite words, these automata are called max-automata,
deterministic automata with counters whose acceptance conditions are a boolean
combination of conditions~\myquot{counter~$c$ is bounded during the run}. While
processing the input, a counter may be incremented, reset to zero, or the
maximum of two counters may be assigned to it (hence the name max-automata). In
this work, we investigate delay games with winning conditions given by
max-automata, so-called max-regular conditions.

\paragraph{Our Contribution.} We prove the analogue of the Hosch-Landweber Theorem for
max-regular winning conditions: it is decidable whether the delaying player has
a winning strategy with bounded lookahead. Furthermore, we obtain a
doubly-exponential upper bound on the necessary lookahead, if this is the case.
Finally, we present a max-regular delay game such that the delaying player wins the game, but only with unbounded lookahead. Thus, unlike for $\omega$-regular conditions, bounded lookahead is not sufficient for max-regular conditions. These are, to the best of our knowledge, the first results on delay games with quantitative winning conditions.

WMSO$+$U is able to express many quantitative winning conditions studied in the literature, e.g., winning conditions in parameterized temporal logics like Prompt-LTL~\cite{KupfermanPitermanVardi09}, Parametric LTL~\cite{Zimmermann13}, or Parametric LDL~\cite{FaymonvilleZimmermann14}, finitary parity and Streett games~\cite{ChatterjeeHenzingerHorn09}, and parity and Streett games with costs~\cite{FZ14}. Thus, for all these conditions we can decide whether Player~$O$ wins a delay game with bounded lookahead.

Our proof consists of a reduction to a delay-free game with a max-regular winning condition. Such games can be solved by expressing them as a satisfiability
problem for WMSO$+$UP on infinite trees: the strategy of one player is an
additional labeling of the tree and a path quantifier is able to range over all
strategies of the opponent\footnote{See Example~1 in~\cite{Bojanczyk14} for
more details.}. The reduction itself is an extension of the one used in the $\exptime$-algorithm for delay games with $\omega$-regular winning
conditions~\cite{KleinZimmermann15} and is based on an equivalence relation that captures the behavior of the automaton recognizing the winning condition. However, unlike the relation used for $\omega$-regular conditions, ours is only correct if applied to words of bounded lengths. Thus, we can deal with bounded lookahead, but not with arbitrary lookahead.

\section{Definitions}
\label{sec_defs}
The set of non-negative integers is denoted by $\nats$. An alphabet $\Sigma$ is a non-empty finite set of letters, and $\Sigma^{*}$ ($\Sigma^n$, $\Sigma^{\omega}$) denotes the set of finite words (words of length $n$, infinite words) over $\Sigma$. The empty word is denoted by $\varepsilon$, the length of a finite word~$w$ by $|w|$. For $w\in \Sigma^{*}\cup\Sigma^{\omega}$ we write $w(n)$ for the $n$-th letter  of $w$.

\paragraph{Automata.}
Given a finite set~$C$ of counters storing non-negative integers,
\[
\ops(C) = \set{\inc{c}, \reset{c},  \maxx{c}{c_0}{c_1} \mid c,c_0, c_1 \in C}
\]
is the set of counter operations over $C$. A counter valuation over $C$ is a mapping~$\val\colon C \rightarrow \nats$. By $\val \pi$ we denote the counter valuation that is obtained by applying a finite sequence~$\pi \in \ops(C)^*$ of counter operations to $\val$, which is defined as implied by the operations' names.

A max-automaton~$\aut = (Q, C, \Sigma, q_I, \delta, \ell, \phi)$ consists of a finite set~$Q$ of states, a finite set~$C$ of counters, an input alphabet~$\Sigma$, an initial state~$q_I$, a (deterministic and complete) transition function~$\delta \colon Q \times \Sigma \rightarrow Q$, a transition labeling\footnote{Here, and later whenever convenient, we treat $\delta$ as relation~$\delta \subseteq Q \times \Sigma \times Q$.}~$\ell \colon \delta \rightarrow \ops(C)^*$ which labels each transition by a (possibly empty) sequence of counter operations, and an acceptance condition~$\phi$, which is a boolean formula over~$C$.

A run of $\aut$ on $\alpha \in \Sigma^\omega$ is an infinite sequence
\begin{equation}
\label{eq_run}
\rho = (q_0, \alpha(0), q_1) \,  (q_1, \alpha(1), q_2) \,  (q_2, \alpha(2), q_3) \cdots  \in \delta^\omega
\end{equation}
with $q_0 = q_I$. Partial (finite) runs on finite words are defined analogously, i.e., $(q_0, \alpha(0), q_1) \cdots  (q_{n-1}, \alpha(n-1), q_n)$ is the run of $\aut$ on $\alpha(0) \cdots \alpha(n-1)$ starting in $q_0$. We say that this run ends in $q_n$. As $\delta$ is deterministic, $\aut$ has a unique run on every finite or infinite word.

Let $\rho$ be as in~(\ref{eq_run}) and define $\pi_n = \ell(q_n, \alpha(n), q_{n+1})$, i.e., $\pi_n$ is the label of the $n$-th transition of $\rho$. Given an initial counter valuation~$\val$ and a counter $c \in C$, we define the sequence
\[ \rho_c = \val(c)\, , \, \val\pi_0(c)\, ,\,  \val\pi_0 \pi_1(c)\,  ,\,  \val\pi_0\pi_1\pi_2(c)\,,  \ldots  \]
of counter values of $c$ reached on the run after applying \emph{all} operations of a transition label.
The run~$\rho$ of $\aut$ on $\alpha$ is accepting, if the acceptance condition~$\phi$ is satisfied by the variable valuation that maps a counter~$c$ to true if and only if $\limsup \rho_c$ is finite. Thus, $\phi$ can intuitively be understood as a boolean combination of conditions~\myquot{$\limsup \rho_c < \infty$}. Note that the limit superior of $\rho_c$ is independent of the initial valuation used to define $\rho_c$, which is the reason it is not part of the description of $\aut$. We denote the language accepted by $\aut$ by $L(\aut)$ and say that it is max-regular.

A parity condition (say min-parity) can be expressed in this framework using a counter for each color that is incremented every time this color is visited and employing the acceptance condition to check that the smallest color whose associated counter is unbounded, is even. Hence, the class of $\omega$-regular languages is contained in the class of max-regular languages.

Given an automaton~$\aut$ over $\SigmaI\times\SigmaO$, we denote by $\proj{\aut}$ the automaton obtained by projecting each letter to its first component, which recognizes the projection of $L(\aut)$ to $\SigmaI$. 



\paragraph{Games with Delay.}
A delay function is a mapping $f\colon\nats\rightarrow \nats \setminus \set{0}$, which is said to be constant, if $f(i)=1$ for every $i > 0$. Given a delay function~$f$ and an $\omega$-language $L\subseteq \left(\SigmaI\times\SigmaO\right)^\omega$, the game $\delaygame{L}$ is played by two players (Player~$I$ and Player~$O$) in rounds $i=0,1,2,\ldots$ as follows: in round $i$, Player~$I$ picks a word $u_i\in\SigmaI^{f(i)}$, then Player~$O$ picks one letter $v_i\in\SigmaO$. We refer to the sequence $(u_0,v_0),(u_1,v_1),(u_2,v_2),\ldots$ as a play of $\delaygame{L}$, which yields two infinite words $\alpha=u_0u_1u_2\cdots$ and $\beta=v_0v_1v_2\cdots$. Player~$O$ wins the play if and only if the outcome~${\alpha(0)\choose \beta(0)}{\alpha(1)\choose \beta(1)}{\alpha(2)\choose \beta(2)}\cdots$ is in $L$, otherwise Player~$I$ wins.

Given a delay function $f$, a strategy for Player~$I$ is a mapping $\stratI\colon \SigmaO^*\rightarrow \SigmaI^*$ such that $|\stratI(w)|=f(|w|)$, and a strategy for Player~$O$ is a mapping $\stratO\colon \SigmaI^*\rightarrow \SigmaO$. Consider a play $(u_0,v_0),(u_1,v_1),(u_2,v_2),\ldots$ of $\delaygame{L}$. Such a play is consistent with $\stratI$, if $u_{i}=\stratI(v_0\cdots v_{i-1})$ for every $i$; it is consistent with $\stratO$, if $v_i=\stratO(u_0\cdots u_i)$ for every $i$. A strategy $\tau$ for Player~$\p$ is winning for her, if every play that is consistent with $\tau$ is won by Player~$\p$. In this case, we say Player~$p$ wins $\delaygame{L}$. A delay game is determined, if one of the players has a winning strategy.

\begin{theorem}
Delay games with max-regular winning conditions are determined.	
\end{theorem}

\begin{proof}
We model a delay game~$\delaygame{L(\aut)}$ for a max-automaton $\aut$ as a parity game\footnote{\label{gtw}See, e.g., \cite{GraedelThomasWilke02} for a detailed definition of parity games.}~$\game$ with finitely many colors in a countable arena. As such games are determined~\cite{EmersonJutla91,Mostowski91}, so is $\delaygame{L(\aut)}$.

A vertex of the parity game stores the round number~$i \in \nats$, an indicator~$t \in \set{I,O}$ which denotes whose player's turn it is, a state~$q$ of $\aut$, and the current lookahead~$w \in \SigmaI^*$. Furthermore, it stores the current counter valuation~$\val_{\mathrm{cur}}$ and a counter valuation~$\val_{\mathrm{max}}$ keeping track of the maximal value that a counter has assumed thus far. Finally, for every counter~$c$ there is a boolean flag~$u_c$ that is set to true if the value~$\val_{\mathrm{max}}$ is updated.

The successors of a vertex of the form~$(i, I, q, w, \val_{\mathrm{cur}},  \val_{\mathrm{max}}, (u_c)_{c\in C})$,
i.e., it is Player~$I$'s turn,  have the form~$(i, O, q, ww', \val_{\mathrm{cur}},  \val_{\mathrm{max}}, (u_c)_{c\in C})$
for some $w' \in \SigmaI^{f(i)}$, i.e., Player~$I$ makes his move in round~$i$ by picking some $w' \in \SigmaI^{f(i)}$, which is appended to the current lookahead~$w$. Dually, successors of a vertex of the form~$(i, O, q, aw, \val_{\mathrm{cur}}, \val_{\mathrm{max}}, (u_c)_{c\in C})$
for $a \in \SigmaI$ and $w \in \SigmaI^*$, i.e., it is Player~$O$'s turn, have the form~$(i+1, I, \delta(q, \textstyle{a \choose b}), w, \val_{\mathrm{cur}}',  \val_{\mathrm{max}}', (u_c')_{c\in C})$
for some $b \in \SigmaO$, where we have $\pi = \ell(q, {a \choose b}, \delta(q, {a \choose b}))$, $\val_{\mathrm{cur}}' = \val_{\mathrm{cur}}\pi$, $\val_{\mathrm{max}}'(c) = \max(\val_{\mathrm{cur}}'(c), \val_{\mathrm{max}}(c))$, and $u_c' = 1$ if and only if $\val_{\mathrm{max}}'(c) >  \val_{\mathrm{max}}(c)$. Here, Player~$O$ makes her move in round~$i$, which consists of picking a letter~$b$. The state of $\aut$,  the variable valuations, and the flags are updated accordingly.

A play is winning for Player~$O$, if the set of counters~$c$ whose flag~$u_c$ is set to~$1$ infinitely often, satisfies the winning condition~$\phi$. This is a Muller condition\footnote{Again, see \cite{GraedelThomasWilke02} for a formal definition.} defined on a finite set of colors, namely the powerset of the set of counters. Applying the LAR-reduction turns the Muller condition into a parity condition with finitely many colors while keeping the arena countable, which yields the parity game~$\game$.

Let $\val_0$ map every counter to $0$. Player~$p$ has a winning strategy for the delay game~$\delaygame{L(\aut)}$ if and only if she has a winning strategy for the parity game~$\game$ from the initial vertex~$(0, I, q_I, \epsilon, \val_0, \val_0, (0)_{c \in C})$. Thus, $\delaygame{L(\aut)}$ is determined, as $\game$ is determined.
\qed
\end{proof}

This result is also implied by a recent more general  determinacy theorem for delay games with Borel winning conditions~\cite{KleinZimmermann15b}.


\section{An Equivalence Relation for Max-Automata}
\label{sec_equiv}
Fix $\aut = (Q, C, \Sigma, q_I, \delta, \ell, \phi)$. We use notions introduced in \cite{Bojanczyk11} to define equivalence relations over sequences of counter operations and over words over $\Sigma$ that capture the behavior of $\aut$. To this end, we need to introduce some notation to deal with runs of $\aut$. Given a state~$q$ and $w \in \Sigma^* \cup \Sigma^\omega$, let $\run(q, w)$ be the run of $\aut$ on $w$ starting in $q$. If $w$ is finite, then $\delta^*(q, w)$ denotes the state $\run(q,w)$ ends in. The transition profile of $w \in \Sigma^*$ is the mapping~$q \mapsto \delta^*(q,w)$. 

First, we define inductively what it means for a sequence~$\pi \in \ops(C)^*$ to transfer a counter~$c$ to a counter~$d$. The empty sequence and the operation~$\inc{c}$ transfer every counter to itself. The operation~$\reset{c}$ transfers every counter but $c$ to itself and the operation~$\maxx{c}{c_0}{c_1}$ transfers every counter but $c$ to itself and transfers $c_0$ and $c_1$ to $c$. Furthermore, if $\pi_0$ transfers $c$ to $e$ and $\pi_1$ transfers $e$ to $d$, then $\pi_0\pi_1$ transfers $c$ to $d$. If $\pi$ transfers $c$ to $d$, then we have $\val \pi (d) \ge \val(c)$ for every counter valuation~$\val$, i.e., the value of $d$ after executing $\pi$ is larger or equal to the value of $c$ before executing $\pi$, independently of the initial counter values.

Furthermore, a sequence of counter operations~$\pi$ transfers $c$ to $d$ with an increment, if there is a counter~$e$ and a decomposition~$\pi_0\,(\inc{e}) \,\pi_1$ of $\pi$ such that $\pi_0$ transfers $c$ to $e$ and $\pi_1$ transfers $e$ to $d$. If $\pi$ transfers $c$ to $d$ with an increment, then we have $\val \pi (d) \ge \val(c) +1$ for every counter valuation~$\val$. 

Finally, we say that $\pi$ is a $c$-trace of length~$m$, if there is a decomposition~$\pi = \pi_0 \cdots \pi_{m-1}$ and a sequence of counters~$c_0, c_1, \ldots, c_m$ with $c_m = c$ such that each $\pi_i$ transfers $c_{i}$ to $c_{i+1}$ with an increment. If $\pi$ is a $c$-trace of length~$m$, then we have $\val \pi (c) \ge m$ for every counter valuation~$\val$.

As only counter values reached after executing all counter operations of a transition label are considered in the semantics of max-automata, we treat $\lbls = \set{\ell(q,a,q') \mid (q, a, q') \in \delta}$ as an alphabet. Every word $\lambda \in \lbls^*$ can be flattened to a word in $\ops(C)^*$, which is denoted by $\flatten(\lambda)$. However, infixes, prefixes, or suffixes of $\lambda$ are defined with respect to the alphabet $\lbls$. We define $\ell(q, w) \in \lbls^*$ to be the sequence of elements in $\lbls$ labeling the run~$\run(q, w)$.

Let $\rho$ be a finite run of $\aut$ and let $\pi \in \ops(C)^*$. We say that $\rho$ ends with $\pi$, if $\pi$ is a suffix of $\flatten(\ell(\rho))$. A finite or infinite run contains $\pi$, if it has a prefix that ends in $\pi$.

\begin{lemma}[\cite{Bojanczyk11}]
\label{lemma_ctraces}
Let $\rho$ be a run of $\aut$ and $c$ a counter. Then, $\limsup \rho_c = \infty$ if and only if $\rho$ contains arbitrarily long $c$-traces.
\end{lemma}

We use the notions of transfer (with increment) to define the equivalence relations that capture $\aut$'s behavior\footnote{In the conference version of this paper~\cite{Zimmermann15}, the third requirement is missing, which invalidates the proof of Lemma~\ref{lemma_swapeqwords}. The correction presented here has only a small influence on the bounds presented in Remark~\ref{remark_indexbounds} and in Corollary~\ref{cor_delaybounds}, but the statement of Lemma~\ref{lemma_swapeqwords} and thus also the proof of our main theorem are not affected.}. We say that $\lambda, \lambda' \in \lbls^*$ are equivalent, denoted by $\lambda \eqop \lambda'$, if for all counters~$c$ and $d$: 
\begin{enumerate}

	\item the flattening of $\lambda$ transfers $c$ to $d$ if and only if the flattening of $\lambda'$ transfers $c$ to $d$, 

	\item the flattening of $\lambda$ transfers $c$ to $d$ with an increment if and only if the flattening of $\lambda'$ transfers $c$ to $d$ with an increment, and

	\item $\lambda$ has a prefix whose flattening transfers $c$ to $d$  if and only if $\lambda'$ has a prefix whose flattening transfers $c$ to $d$.

\end{enumerate}
Using this, we define two words $x,x' \in \Sigma^*$ to be equivalent, denoted by $x \eqword x'$, if they have the same transition profile and if $\ell(q, x)\eqop \ell(q, x')$ for all states $q$. 

\begin{remark}
\label{remark_indexbounds}
Let $\aut$ be a max-automaton with $n$ states and $k$ counters.
\begin{enumerate}
	\item The index of $\eqop$ is at most $2^{k^2}$.
	\item The index of $\eqword$ is at most $2^{n(\log n + 6k^2)}$.
\end{enumerate}
\end{remark}

Next, we show that we can decompose an infinite word~$\alpha$ into $x_0 x_1 x_2 \cdots$ and replace each $x_i$ by an $\eqword$-equivalent $x_i'$ without changing membership in $L(\aut)$, provided the lengths of the $x_i$ and the lengths of the $x_i'$ are bounded.

\begin{lemma}
	\label{lemma_swapeqwords}
Let $(x_i)_{i \in \nats}$ and $(x_i')_{i \in \nats}$ be two sequences of words over $\Sigma^*$ with $\sup_i |x_i| < \infty$, $\sup_i |x_i'| < \infty$, and $x_i \eqword x_i'$ for all $i$. Then, $x = x_0 x_1 x_2 \cdots \in L(\aut)$ if and only if $x' = x_0' x_1' x_2' \cdots \in L(\aut)$.
\end{lemma}

\begin{proof}
Let $\rho$ and $\rho'$ be the run of $\aut$ on $x$ and $x'$, respectively.
We show that $\rho$ contains arbitrarily long $c$-traces if and only if $\rho'$ contains arbitrarily long $c$-traces. Due to Lemma~\ref{lemma_ctraces}, this suffices to show that the run of $\aut$ on $x$ is accepting if and only if the run of $\aut$ on $x'$ is accepting. Furthermore, due to symmetry, it suffices to show one direction of the equivalence. Thus, assume $\rho$ contains arbitrarily long $c$-traces and pick $m' \in \nats$ arbitrarily. We show the existence of a $c$-trace of length~$m'$ contained in $\rho'$. To this end, we take a $c$-trace in $\rho$ of length $m > m'$ for some sufficiently large $m$ and show that the $\eqop$-equivalent part of $\rho'$ contains a $c$-trace of length~$m'$. 
	
By definition of $\eqword$, processing $x_0 \cdots x_{i-1}$ and processing $x_0' \cdots x_{i-1}'$ brings $\aut$ to the same state, call it $q_i$. Furthermore, let $\pi_i = \ell(q_i, x_i)$ be the sequence of counter operations labeling the run of $\aut$ on $x_i$ starting in $q_i$, which ends in $q_{i+1}$. The sequences $\pi_i'$ labeling the runs on the $x_i'$ are defined analogously. By $x_i \eqword x_i'$ we conclude that $\pi_i$ and $\pi_i'$ are $\eqop$-equivalent as well. Furthermore, define $b = \sup_i |x_i|$, which is well-defined due to our assumption, and define $m = (m'+1) \cdot o \cdot b $, where $o = \max_{\pi \in \Lambda}\size{\pi}$ is the maximal length of a sequence of operations labeling a transition. Each $\pi_i$ can contribute at most $|\pi_i|$ increments to a $c$-trace that subsumes $\pi_i$, which is bounded by $|\pi_i| \le o \cdot b$.

Now, we pick $i$ such that $\pi_0 \cdots \pi_i$ contains a $c$-trace of length~$m$. We can assume w.l.o.g.\ that the trace starts at the beginning of $\pi_s$ for some $s \le i$ and ends in a prefix of $\pi_i$. Hence, there are counters~$c_s, c_{s+1}, \ldots, c_{i}$ such that the flattening of $\pi_{j}$ transfers $c_j$ to $c_{j+1}$ for every $j$ in the range~$s \le j < i$ and that $\pi_i$ has a prefix whose flattening transfers $c_{i}$ to $c$. 
Furthermore, by the choice of $m$ we know that at least $m'$ of these transfers are actually transfers with increments, as every transfer contains at most $b \cdot o$ increments.

The equivalence of $\pi_j$ and $\pi_j'$ implies that $\pi_j'$ realizes  the same transfers (with increments) as $\pi_j$. Hence, $\pi_0' \cdots \pi_i'$ contains a $c$-trace of length~$m'$ as well. \qed
\end{proof}

Note that the lemma does not hold if we drop the boundedness requirements on the lengths of the $x_i$ and the $x_i'$.

To conclude, we show that the equivalence classes of $\eqword$ are regular and can be \textit{tracked} on-the-fly by a finite automaton~$\auttrack$ in the following sense.

\begin{lemma}
\label{lemma_eqclasstracker}
There is a deterministic finite automaton~$\auttrack$ with set of states~${\Sigma\quotient}$ such that the run of $\auttrack$ on $w \in \Sigma^*$ ends in $[w]_{\eqword}$.
\end{lemma}

\begin{proof}
Define $\auttrack = (\Sigma\quotient, \Sigma, \eqclassword{\epsilon}, \delta_{\auttrack}, \emptyset)$ where 
$\delta_{\auttrack}(\eqclassword{x}, a) = \eqclassword{xa}$,
which is independent of the representative~$x$ and based on the fact that $\eqop$ (and thus also $\eqword$) is a congruence, i.e., $\pi_0 \eqop \pi_1$ implies $\pi_0 \pi \eqop \pi_1\pi$ for every $\pi$. A straightforward induction over $|w|$ shows that $\auttrack$ has the desired properties.
\qed
\end{proof}

\begin{corollary}
\label{coro_eqclassesregular}
Every $\eqword$-equivalence class is regular.	
\end{corollary}


\section{Reducing Delay Games to Delay-free Games}
\label{sec_mainthm}
In this section, we prove our main theorem.

\begin{theorem}
	\label{theorem_main}
The following problem is decidable: given a max-automaton~$\aut$, does Player~$O$ win $\delaygame{L(\aut)}$ for some constant delay function~$f$?
\end{theorem}

To prove this result, we construct a delay-free game in a finite arena with a max-regular winning condition that is won by Player~$O$ if and only if she wins $\delaygame{L(\aut)}$ for some constant delay function~$f$. The winner of such a game can be determined effectively.

Let $\aut = (Q, C, \SigmaI \times \SigmaO, q_I, \delta, \ell, \phi)$ and let  
$\auttrack = ((\SigmaI \times \SigmaO)\quotient, \SigmaI \times \SigmaO, [\epsilon]_{\eqword}, \delta_{\auttrack}, \emptyset)$
be defined as in Lemma~\ref{lemma_eqclasstracker}. For the sake of readability, we denote the $\eqword$-equivalence class of $w$ by $\eqclass{w}$ without a subscript. Furthermore, we denote equivalence classes using the letter $S$.
We define the product~$\autproduct = (Q_\autproduct, C, \SigmaI \times \SigmaO, q_I^\autproduct, \delta_\autproduct, \ell_\autproduct, \phi)
$ of $\aut$ and $\auttrack$, which is a max-automaton, where
\begin{itemize}
	
	\item $Q_\autproduct = Q \times ((\SigmaI \times \SigmaO)\quotient)$,
	
	\item $q_I^\autproduct = (q_I, [\epsilon]_{\eqword})$,
	
	\item $\delta_\autproduct((q,S),a) = (\delta(q, a),\delta_{\auttrack}(S,a))$ for a states $q \in Q$, an equivalence class~$S \in (\SigmaI\times\SigmaO)\quotient$, and a letter~$a \in \SigmaI \times \SigmaO$, and
	
	\item $\ell_\autproduct((q, S), a, (q',S')) = \ell(q, a, q')$.
	
\end{itemize} 
Let $n = |Q_\autproduct|$. We have $L(\autproduct) = L(\aut)$, since acceptance only depends on the component~$\aut$ of $\autproduct$. However, we are interested in partial runs of $\autproduct$, as the component~$\auttrack $ keeps track of the equivalence class of the input processed by~$\autproduct$.

\begin{remark}
\label{remark_autp}
Let $w \in (\SigmaI \times \SigmaO)^*$ and let $(q_0, S_0)(q_1, S_1) \cdots (q_{|w|}, S_{|w|})$ be the run of $\autproduct$ on $w$ from some state~$(q_0, S_0)$ with $S_0 = \eqclass{\epsilon}$. Then, $q_0 q_1 \cdots q_{|w|}$ is the run of $\aut$ on $w$ starting in $q_0$ and $S_{|w|} = \eqclass{w}$.
\end{remark}

In the following, we will work with partial functions~$\resolve$ from $Q_\autproduct$ to $\pow$, where we denote the domain of $\resolve$ by $\dom(r)$. Intuitively, we use such a function to capture the information encoded in the lookahead provided by Player~$I$. Assume Player~$I$ has picked $\alpha(0) \cdots \alpha(j)$ and Player~$O$ has picked $\beta(0) \cdots \beta(i)$ for some $i < j$, i.e., the lookahead is $\alpha(i+1) \cdots \alpha(j)$. Then, we can determine the state~$q$ that $\autproduct$ reaches when processing ${\alpha(0) \choose \beta(0)} \cdots {\alpha(i) \choose \beta(i)}$, but the automaton cannot process $\alpha(i+1) \cdots \alpha(j)$, since Player~$O$ has not yet provided her moves $\beta(i+1) \cdots \beta(j)$. However, we can determine which states Player~$O$ can enforce by picking an appropriate completion. These will be contained in $\resolve(q)$. 

To formalize this, we use the function~$\delta_\autpow \colon 
\pow \times \SigmaI \rightarrow \pow$ defined via
$ \delta_\autpow(P, a) = \bigcup_{q \in P} \bigcup_{b \in \SigmaO}  \delta_\autproduct\left(q, {a \choose b}\right) $,
i.e., $\delta_\autpow$ is the transition function of the powerset automaton of the projection automaton~$\proj{\autproduct}$.
As usual, we extend $\delta_\autpow$ to $\delta_\autpow^* \colon \pow \times \SigmaI^* \rightarrow \pow $ via $\delta_\autpow^*(P, \epsilon) = P$ and
$\delta_\autpow^*(P, wa) = \delta_\autpow( \delta_\autpow^*(P, w) , a)$.

Let $D \subseteq Q_\autproduct$ be non-empty and let $w \in \SigmaI^*$. We define the function~$\resolve_w^D$ with domain $D$ as follows: for every $(q,S) \in D$, we have 
\[ \resolve_w^D(q,S) = \delta_\autpow^*(\set{(q,\eqclass{\epsilon})},w) \enspace, \]
i.e., we collect all states~$(q',S')$ reachable from $(q,\eqclass{\epsilon})$ (note that the second component is the equivalence class of the empty word, not the class~$S$ from the argument) via a run of $\proj{\autproduct}$ on $w$. Thus, if $(q', S') \in \resolve_w^D(q,S)$, then there is a word~$w'$ whose projection is $w$ and with $\eqclass{w'}=S'$ such that the run of $\aut$ on $w'$ leads from $q$ to $q'$. Thus, if Player~$I$ has picked the lookahead~$w$, then Player~$O$ could pick an answer such that the combined word leads $\aut$ from $q$ to $q'$ and such that it is a representative of $S'$. 
 
We call $w$ a witness for a partial function $\resolve \colon Q_\autproduct \rightarrow \pow$, if we have $\resolve = \resolve_w^{\dom(r)}$. Thus, we obtain a language~$\wit{\resolve} \subseteq \SigmaI^*$ of witnesses for each such function~$\resolve$. Now, we define $\curlyR = \set{\resolve \mid \dom(\resolve)\neq \emptyset \text{ and } \wit{\resolve} \text{ is infinite}}$.

\begin{lemma}
\label{lem_graphautprops}
Let $\curlyR$ be defined as above.
\begin{enumerate}
	
	\item\label{lem_graphautprops_witnonempty}
	Let $\resolve \in \curlyR$. Then, $\resolve(q) \not= \emptyset$ for every $q \in \dom(\resolve)$.
	
	\item 
	Let $\resolve$ be a partial function from $Q_\autproduct$ to $\pow$. Then, $\wit{\resolve}$ is recognized by a deterministic finite automaton with $2^{n^2}$ states.
	
	\item\label{lem_graphautprops_witnessdense} Let $\resolve \in  \curlyR$. Then, $\wit{\resolve}$ contains a word~$w$ with $k \le |w| \le k + 2^{n^2}$ for every $k$.
		
	\item\label{lem_graphautprops_witdisjoint}
	Let $\resolve \neq \resolve' \in \curlyR$ such that  $\dom(\resolve) = \dom(\resolve')$. Then, $\wit{\resolve} \cap \wit{\resolve'} = \emptyset$.
	
	\item\label{lem_graphautprops_witcomplete}
	Let $D \subseteq Q_\autproduct$ be non-empty and let $w$ be such that $|w| \ge 2^{n^2}$. Then, there exists some $\resolve \in \curlyR$ with $\dom(\resolve) = D$ and $w \in \wit{\resolve}$.
	
\end{enumerate}
\end{lemma}

Due to items~\ref{lem_graphautprops_witdisjoint}.) and \ref{lem_graphautprops_witcomplete}.), we can define for every non-empty~$D \subseteq Q_\autproduct$ a function $\witmap{D}$ that maps words~$w \in \SigmaI^*$ with $|w| \ge 2^{n^2}$ to the unique function~$\resolve$ with $\dom(\resolve) = D$ and $w \in \wit{\resolve}$. This will be used later in the proof.

Now, we define an abstract game~$\game(\aut)$ between Player~$I$ and Player~$O$ that is played in rounds $i = 0, 1, 2, \ldots$: in each round, Player~$I$ picks a function from $\curlyR$ and then Player~$O$ picks a state~$q$ of $\autproduct$. In round~$0$, Player~$I$ has to pick $\resolve_0$ subject to constraint~(C1): $\dom(\resolve_0) = \set{q_I^\autproduct}$. Then, Player~$O$ has to pick a state~$q_0 \in \dom(\resolve_0)$ (which implies $q_0 = q_I^\autproduct$). Now, consider round $i>0$: Player~$I$ has picked functions~$\resolve_0, \resolve_1, \ldots, \resolve_{i-1}$ and Player~$O$ has picked states~$q_0, q_1, \ldots, q_{i-1} $. Now, Player~$I$ has to pick a function~$\resolve_i$ subject to constraint~(C2): $\dom(\resolve_i) = \resolve_{i-1}(q_{i-1}).$ 
Then, Player~$O$ has to pick a state~$q_i \in \dom(\resolve_i)$. Both players can always move: Player~$I$ can, as $\resolve_{i-1}(q_{i-1})$ is always non-empty (Lemma~\ref{lem_graphautprops}.\ref{lem_graphautprops_witnonempty}) and thus the domain of some $\resolve \in \curlyR$ (Lemma~\ref{lem_graphautprops}.\ref{lem_graphautprops_witcomplete}) and Player~$O$ can, as the domain of every~$\resolve \in \curlyR$ is non-empty by construction.

The resulting play is the sequence~$ \resolve_0 q_0 \resolve_1 q_1 \resolve_2 q_2\cdots $. Let $q_i = (q_i', S_i)$ for every $i$, i.e., $S_i$ is an $\eqword$-equivalence class. Let $x_i \in S_i$ for every $i$ such that $\sup_{i}|x_i|< \infty$. Such a sequence can always be found as $\eqword$ has finite index. Player~$O$ wins the play if the word~$x_0 x_1 x_2 \cdots$ is accepted by $\aut$. Due to Lemma~\ref{lemma_swapeqwords}, this definition is independent of the choice of the representatives~$x_i$. 

A strategy for Player~$I$ is a function~$\stratI'$ mapping the empty play prefix to a function~$\resolve_0$ subject to constraint~(C1) and mapping a non-empty play prefix~$\resolve_0 q_0 \cdots \resolve_{i-1} q_{i-1} $ ending in a state to a function~$\resolve_i$ subject to constraint~(C2). On the other hand, a strategy for Player~$O$ maps a play prefix~$\resolve_0 q_0 \cdots  \resolve_{i} $ ending in a function to a state~$q_i\in \dom(\resolve_{i})$. A play $ \resolve_0 q_0 \resolve_1 q_1 \resolve_2 q_2 \cdots $ is consistent with $\stratI'$, if $\resolve_i = \stratI'(\resolve_0 q_0 \cdots \resolve_{i-1} q_{i-1} )$ for every $i \ge 0$. Dually, the play is consistent with $\stratO'$, if $q_i = \stratO'(\resolve_0 q_0 \cdots  \resolve_i)$ for every $i \ge 0$. A strategy is winning for Player~$\p$, if every play that is consistent with this strategy is winning for her. As usual, we say that Player~$\p$ wins $\game(\aut)$, if she has a winning strategy. 

\begin{lemma}
	\label{lem_splitgamecorrectness}
Player~$O$ wins $\delaygame{L(\aut)}$ for some constant delay function~$f$ if and only if Player~$O$ wins $\game(\aut)$.
\end{lemma}

\begin{proof}
First, assume Player~$O$ has a winning strategy~$\stratO$ for $\delaygame{L(\aut)}$ for some constant delay function~$f$. We construct a winning strategy~$\stratO'$ for Player~$O$ in $\game(\aut)$ via simulating a play of $\game(\aut)$ by a play of $\delaygame{L(\aut)}$.

Let $\resolve_0$ be the first move of Player~$I$ in $\game(\aut)$, which has to be responded to by Player~$O$ by picking $q_I^\autproduct = \stratO'(\resolve_0)$, and let $\resolve_1$ be Player~$I$'s response to that move. Let $w_0 \in \wit{\resolve_0}$ and $w_1\in \wit{\resolve_1}$ be witnesses for the functions picked by Player~$I$. Due to Lemma~\ref{lem_graphautprops}.\ref{lem_graphautprops_witnessdense}, we can choose $w_0$ and $|w_1|$ with $f(0) \le |w_0|, |w_1| \le f(0) + 2^{n^2}$. 

We simulate the play prefix~$\resolve_0 q_0 \resolve_1$ in $\delaygame{L(\aut)}$, where $q_0 = q_I^\autproduct$: Player~$I$ picks $w_0 w_1 = \alpha(0) \cdots \alpha(\ell_1 -1)$ in his first moves and let $\beta(0) \cdots \beta(\ell_1 - f(0))$ be the response of Player~$O$ according to $\stratO$. We obtain $|\beta(0) \cdots \beta(\ell_1 - f(0))| \ge |w_0|$, as $|w_1| \ge f(0)$. 

Thus, we are in the following situation for $i=1$: in $\game(\aut)$, we have a play prefix~$\resolve_0 q_0 \cdots \resolve_{i-1} q_{i-1} \resolve_i$ and in $\delaygame{L(\aut)}$, Player~$I$ has picked 
$w_0 w_1 \cdots w_i = \alpha(0) \cdots \alpha(\ell_i -1)$ and Player~$O$ has picked $\beta(0) \cdots \beta(\ell_i - f(0))$ according to $\stratO$, where $|\beta(0) \cdots \beta(\ell_i - f(0))| \ge |w_0 \cdots w_{i-1}|$. Furthermore, $w_j$ is a witness for $r_j$ for every $j \le i$.

In this situation, let $q_i$ be the state of $\autproduct$ that is reached when processing $w_{i-1}$ and the corresponding moves of Player~$O$, i.e., 
\[{\alpha(|w_0 \cdots w_{i-2}|) \choose \beta(|w_0 \cdots w_{i-2}|)} \cdots {\alpha(|w_0 \cdots w_{i-1}|-1) \choose \beta(|w_0 \cdots w_{i-1}|-1)} \enspace,\]
 starting in state~$(q_{i-1}', \eqclass{\epsilon})$, where $q_{i-1} = (q_{i-1}', S_{i-1})$.
 
By definition of $\resolve_{i-1}$, we have $q_i \in \resolve_{i-1}(q_{i-1})$, i.e., $q_i$ is a legal move for Player~$O$ in $\game(\aut)$ to extend the play prefix~$\resolve_0 q_0 \cdots \resolve_{i-1} q_{i-1} \resolve_i$. Thus, we define $\stratO'(\resolve_0 q_0 \cdots \resolve_{i-1} q_{i-1} \resolve_i) = q_i$. Now, let $\resolve_{i+1}$ be the next move of Player~$I$ in $\game(\aut)$ and let $w_{i+1} \in \wit{\resolve_{i+1}}$ be a witness with $f(0) \le |w_{i+1}| \le f(0) + 2^{n^2}$. Going back to $\delaygame{L(\aut)}$, let Player~$I$ pick $w_{i+1} = \alpha(\ell_i) \cdots \alpha(\ell_{i+1} - 1)$ as his next moves and let $\beta(\ell_i -f(0)+1) \cdots \beta(\ell_{i+1} - f(0))$ be the response of Player~$O$ according to $\stratO$. Then, we are in the situation as described in the previous paragraph, which concludes the definition of $\stratO'$.

It remains to show that $\stratO'$ is winning for Player~$O$ in $\game(\aut)$. Consider a play~$\resolve_0 q_0 \resolve_1 q_1 \resolve_2 q_2 \cdots$ that is consistent with $\stratO'$ and let $w = {\alpha(0) \choose \beta(0)}{\alpha(1) \choose \beta(1)}{\alpha(2) \choose \beta(2)} \cdots$ be the corresponding outcome constructed as in the simulation described above. Let $q_i = (q_i', S_i)$, i.e., $q_i'$ is a state of our original automaton~$\aut$. A straightforward inductive application of Remark~\ref{remark_autp} shows that $q_i'$ is the state that $\aut$ reaches after processing $w_{i}$ and the corresponding moves of Player~$O$, i.e.,
\[x_i = {\alpha(|w_0 \cdots w_{i-1}|) \choose \beta(|w_0 \cdots w_{i-1}|)}
 \cdots
{\alpha(|w_0 \cdots w_{i}|-1) \choose \beta(|w_0 \cdots w_{i}|-1)} \enspace, \]
starting in $q_{i-1}'$, and that $S_i = \eqclass{x_i}$. Note that the length of the $x_i$ is bounded, i.e., we have $\sup_i |x_i| \le f(0) + 2^{n^2}$.

As $w$ is consistent with a winning strategy for Player~$O$, the run of $\aut$ on $w = x_0 x_1 x_2 \cdots$ is accepting. Thus, we conclude that the play $\resolve_0 q_0 \resolve_1 q_1 \resolve_2 q_2 \cdots$ is winning for Player~$O$, as the $x_i$ are a bounded sequence of representatives. Hence, $\stratO'$ is indeed a winning strategy for Player~$O$ in $\game(\aut)$.

%

Now, we consider the other implication: assume Player~$O$ has a winning strategy~$\stratO'$ for $\game(\aut)$ and fix $d = 2^{n^2}$. We construct a winning strategy~$\stratO$ for her in $\delaygame{L(\aut)}$ for the constant delay function~$f$ with $f(0) = 2d$. In the following, both players pick their moves in blocks of length~$d$. We denote Player~$I$'s blocks by $a_i$ and Player~$O$'s blocks by $b_i$, i.e., in the following, every $a_i$ is in $\SigmaI^d$ and every $b_i$ is in $\SigmaO^d$. This time, we simulate a play of $\delaygame{L(\aut)}$ by a play in $\game(\aut)$.

Let $a_0 a_1$ be the first move of Player~$I$ in $\delaygame{L(\aut)}$, let $q_0 = q_I^\autproduct$, and define the functions $\resolve_0 = \witmap{\set{q_0}}(a_0)$ and 
$\resolve_1 = \witmap{\resolve_0(q_0)}(a_1)$. Then, $\resolve_0 q_0 \resolve_1$ is a legal play prefix of $\game(\aut)$ that is consistent with the winning strategy~$\stratO'$ for Player~$O$. 

Thus, we are in the following situation for $i=1$: in $\game(\aut)$, we have constructed a play prefix~$\resolve_0 q_0  \cdots \resolve_{i-1} q_{i-1} \resolve_i $ that is consistent with $\stratO'$; in $\delaygame{L(\aut)}$, Player~$I$ has picked $a_0 \cdots a_i$ such that $a_j$ is a witness for $\resolve_j$ for every $j$ in the range~$ 0 \le j \le i$. 
Player~$O$ has picked $b_0 \cdots b_{i-2}$, which is the empty word for $i = 1$. 

In this situation, let $q_i = \stratO'( \resolve_0 q_0  \cdots \resolve_{i-1} q_{i-1} \resolve_i )$. By definition, we have $q_i \in \dom(\resolve_i) = \resolve_{i-1}(q_{i-1})$. Furthermore, as $a_{i-1}$ is a witness for $\resolve_{i-1}$, there exists $b_{i-1}$ such that $\autproduct$ reaches the state~$q_i$ when processing ${ a_{i-1} \choose b_{i-1}}$ starting in state~$(q_{i-1}',\eqclass{\epsilon})$, where $q_{i-1} = (q_{i-1}', S_{i-1})$.
 
Player~$O$'s strategy for $\delaygame{L(\aut)}$ is to play $b_{i-1}$ in the next $d$ rounds, which is answered by Player~$I$ by picking some $a_{i+1}$ during these rounds. This induces the function~$\resolve_{i+1} = \witmap{\resolve_i(q_i)}(a_{i+1})$. Now, we are in the same situation as described in the previous paragraph. This finishes the description of the strategy $\stratO$ for Player~$O$ in $\delaygame{L(\aut)}$.

It remains to show that $\stratO$ is winning for Player~$O$. Let $w = {a_0 \choose b_0}{a_1 \choose b_1}{a_2 \choose b_2}\cdots$ be the outcome of a play in $\delaygame{L(\aut)}$ that is consistent with $\stratO$. Furthermore, let $\resolve_0 q_0 \resolve_1 q_1 \resolve_2 q_2 \cdots $ be the corresponding play in $\game(\aut)$ constructed in the simulation as described above, which is consistent with $\stratO'$. Let $q_i = (q_i', S_i)$. A straightforward inductive application of Remark~\ref{remark_autp} shows that $q_i'$ is the state reached by $\aut$ after processing
$x_i = {a_i \choose b_i}$
 starting in $q_{i-1}'$
and $S_i = \eqclass{x_i}$. Furthermore, we have $\sup_i |x_i| = d$.

As $\resolve_0 q_0 \resolve_1 q_1 \resolve_2 q_2 \cdots $ is consistent with a winning strategy for Player~$O$ and  therefore winning for Player~$O$, we conclude that $x_0 x_1 x_2 \cdots$ is accepted by $\aut$. Hence, $\aut$ accepts the outcome~$w$, which is equal to $x_0 x_1 x_2 \cdots$, i.e., the play in $\delaygame{L(\aut)}$ is winning for Player~$O$. Thus, $\stratO$ is a winning strategy for Player~$O$ in $\delaygame{L(\aut)}$.
\qed
\end{proof}

Now, we can prove our main theorem of this section, Theorem~\ref{theorem_main}.

\begin{proof}	
Due to Lemma~\ref{lem_splitgamecorrectness}, we just have to show that we can construct and solve an explicit version of $\game(\aut)$. To this end, we encode $\game(\aut)$ as a graph-based game with arena~$(V, V_I, V_O, E)$ where
\begin{itemize}
	\item the set of vertices is $V = V_I \cup V_O$ with
	\item the vertices~$V_I = \set{\init} \cup \curlyR \times Q_\autproduct$ of Player~$I$, where $\init$ is a fresh initial vertex,
	\item the vertices~$V_O =  \curlyR $ of Player~$O$, and
	\item $E$ is the union of the following sets of edges:
	\begin{itemize}
		
		\item $\set{(\init, \resolve) \mid \dom(\resolve) = \set{q_I^\autproduct}}$: the initial moves of Player~$I$.
		
		\item $\set{((\resolve, q), \resolve') \mid \dom(\resolve') = \resolve(q)}$: (regular) moves of Player~$I$.
		
		\item $\set{(\resolve,(\resolve, {q})) \mid q \in \dom(\resolve)}$: moves of Player~$O$.
	
	\end{itemize}
	\end{itemize}
	
A play is an infinite path starting in $\init$. To determine the winner of a play, we fix an arbitrary function~$\rep \colon (\SigmaI \times \SigmaO)^*\quotient\, \rightarrow (\SigmaI \times \SigmaO)^*$ that maps each equivalence class to some representative, i.e., $\rep(S) \in S$ for every $S \in (\SigmaI \times \SigmaO)^*\quotient$. Consider an infinite play 
\[
\init, \resolve_0, (\resolve_0, q_0), \resolve_1, (\resolve_1, q_1), \resolve_2, (\resolve_2, q_2), \ldots \enspace,
\] 
with $q_i = (q_i', S_i)$ for every $i$. This play is winning for Player~$O$, if the infinite word~$\rep(S_0)\rep(S_1)\rep(S_2) \cdots$ is accepted by $\aut$ (note that $\sup_i |\rep(S_i)|$ is bounded, as there are only finitely many equivalence classes). The set~$\win \subseteq V^\omega$ of winning plays for Player~$O$ is a max-regular language\footnote{This implies that $\game(\aut)$ is determined, as max-regular conditions are Borel~\cite{Bojanczyk11,Martin75}.}, as it can be recognized by an automaton that simulates the run of $\aut$ on $\rep(S)$ when processing a vertex of the form~$(\resolve, (q,S))$ and ignores all other vertices. Games in finite arenas with max-regular winning condition are decidable via an encoding as a satisfiability problem for WMSO$+$UP~\cite{Bojanczyk14}.

Player~$O$ wins $\game(\aut)$ (and thus $\delaygame{L(\aut)}$ for some constant $f$) if and only if she has a winning strategy from $\init$ in the game~$((V, V_I, V_O, E), \win)$, which concludes the proof.\qed
\end{proof}

We obtain a doubly-exponential upper bound on the constant delay necessary for Player~$O$ to win a delay game with a max-regular winning condition by applying both directions of the equivalence between $\delaygame{\aut}$ and $\game(\aut)$.

\begin{corollary}
\label{cor_delaybounds}
Let $\aut$ be a max-automaton with $n$ states and $k$ counters. The following are equivalent:
\begin{enumerate}
	\item Player~$O$ wins $\delaygame{L(\aut)}$ for some constant delay function~$f$.
	\item Player~$O$ wins $\delaygame{L(\aut)}$ for some constant delay function~$f$ with\newline $f(0) \le 2^{2^{2n(\log n + 6k^2)}+1} $.
\end{enumerate}
\end{corollary}


\section{Constant Delay Does Not Suffice}
\label{sec_linear}
In this section, we show that constant delay does not suffice to win every delay game that Player~$O$ can with with arbitrary delay, i.e., the analogue of the Holtmann-Kaiser-Thomas theorem for delay games with max-regular winning conditions does not hold. In terms of uniformization, the following theorem shows that there are max-regular languages that are uniformizable by (uniformly) continuous functions, but not by Lipschitz-continuous functions.

\begin{theorem}
There is a max-regular language $L$ such that Player~$O$ wins $\delaygame{L}$ for some $f$, but not for any constant~$f$. 	
\end{theorem}

\newcommand{\sepp}{\#}
\newcommand{\blank}{*}

\begin{proof}
Let $\SigmaI = \set{0,1, \sepp}$ and $\SigmaO = \set{0,1,\blank}$. An input block is a word $\sepp w $ with $w \in \set{0,1}^+$, and its length is defined to be $|w|$. An output block is a word 
\[
{\sepp \choose b} 
{\alpha(1) \choose \blank} 
{\alpha(2) \choose \blank} 
\cdots
{\alpha(n) \choose \blank} 
{b \choose b}  \in (\SigmaI \times \SigmaO)^+
\]
for $b \in \set{0,1}$ and $\alpha(j) \in \set{0,1}$ for all $j$ in the range $1 \le j \le n$. The length of the block is $n+1$. Note that the first and last letter in an output block are the only ones whose second component is not a $\blank$, and that these letters have to be equal to the first component of the last letter of the block. Every input block of length $n$ can be extended to an output block of length $n$ and and the projection to the first component of every output block is an input block.

A word~${\alpha(0) \choose \beta(0)}
{\alpha(1) \choose \beta(1)}
{\alpha(2) \choose \beta(2)} \cdots$ over $ \SigmaI \times \SigmaO$ is in $L$ if and only if it satisfies the following property: if $\alpha(0)\alpha(1)\alpha(2) \cdots$ contains infinitely many $\sepp$ and arbitrarily long input blocks, then ${\alpha(0) \choose \beta(0)}
{\alpha(1) \choose \beta(1)}
{\alpha(2) \choose \beta(2)} \cdots$ contains arbitrarily long output blocks.
It is easy to come up with a WMSO$+$U formula defining $L$ by formalizing the definitions of input and output blocks in first-order logic. 

Now, consider $L$ as winning condition for a delay game. Intuitively, Player~$O$ has to specify arbitrarily long output blocks, provided Player~$I$ produces arbitrarily long input blocks. The challenge for Player~$O$ is that she has to specify at the beginning of every output block whether she ends the block in a position where Player~$O$ has picked a $0$ or a $1$. 

First, we show that Player~$O$ wins $\delaygame{L}$ for the delay function~$f$ with $f(i) = 2$ for every $i$. Consider round~$i$ for some $i$ with $\alpha(i) = \sepp$. By the choice of $f$, Player~$O$ has already picked $\alpha(i+1) \cdots \alpha(2i+1)$. Let $j$ in the range $i+1 \le j \le 2i+1$ be maximal with $\alpha(i+1) \cdots \alpha(j) \in \set{0,1}^+$. If $j$ is defined, then Player~$O$ picks $\alpha(j)$ in round~$i$, $\blank$ during rounds~$i+1, \ldots, j-1$, and $\alpha(j)$ in round~$j$. In every other situation, she picks an arbitrary letter. 

Now, consider an outcome~${\alpha(0) \choose \beta(0)}
{\alpha(1) \choose \beta(1)}
{\alpha(2) \choose \beta(2)} \cdots$ that is consistent with this strategy that contains infinitely many $\sepp$ and arbitrarily long input blocks. Let $n \in \nats$ be arbitrary and pick an input block of length at least $n$. We can assume w.l.o.g.\ that the block begins at $\alpha(i)$ with $i > n$. Thus, in round $i$ when Player~$O$ had to pick $\beta(i)$, Player~$I$ had already picked $\alpha(i+1) \cdots \alpha(2i+1)$, which comprises the whole input block of length $n$. Accordingly, Player~$O$ produces an output block of length at least~$n$. Thus, the winning condition is satisfied.

It remains to show that Player~$I$ wins $\delaygame{L}$ for every constant delay function~$f$. He uses a counter~$c$ to produce arbitrarily long input blocks, which is initialized to $f(0)$. In round~$0$, he picks $\sepp 0^{f(0)-1}$. If Player~$O$ answers the $\sepp$ with $0$, then Player~$I$ continues picking $1$ until he has produced an input block of length~$c$. Dually, if Player~$O$ answers the $\sepp$ with $1$ or $\blank$, then Player~$I$ continues picking $0$ until he has produced an input block of length~$c$. In both cases, he continues by incrementing the counter and by picking  $\sepp 0^{f(0)-1}$ during the next rounds until Player~$O$ has to respond to the last $ \sepp$. Then, he continues as in the case distinction described above.

Now, consider an outcome~${\alpha(0) \choose \beta(0)}
{\alpha(1) \choose \beta(1)}
{\alpha(2) \choose \beta(2)} \cdots$ that is consistent with this strategy. It contains infinitely many $\sepp$ and arbitrarily long input blocks. Furthermore, the length of every output block is bounded by $f(0)$, as Player~$I$ is able to react to Player~$O$'s declaration at the beginning of each such block by playing the opposite letter. Thus, the play is winning for Player~$I$.
\qed\end{proof}

We have just shown that Player~$O$ wins the game for the delay function $f(i) =2$. Such a function is called linear~\cite{FridmanLoedingZimmermann11}, as the lookahead grows linearly. 


\section{Conclusion}
\label{sec_conc}
We considered delay games with max-regular winning conditions. Our main result is an algorithm that determines whether Player~$O$ has a winning strategy for some constant delay function, which consists of reducing the original problem to a delay-free game with max-regular winning condition. Such a game can be solved by encoding it as an emptiness problem for a certain class of tree automata (so-called WMSO$+$UP automata) that capture WMSO$+$UP on infinite trees. Our reduction also yields a doubly-exponential upper bound on the necessary constant delay to win such a game, provided Player~$O$ does win for some constant delay function. 

It is open whether the doubly-exponential upper bound is tight. The best lower bounds are exponential and hold already for deterministic reachability and safety automata~\cite{KleinZimmermann15}, which can easily be transformed into max-automata.

We deliberately skipped the complexity analysis of our algorithm, since the reduction of the delay-free game to an emptiness problem for WMSO$+$UP automata does most likely not yield tight upper bounds on the complexity. Instead, we propose to investigate (delay-free) games with max-regular winning conditions, a problem that is worthwhile studying on its own, and to find a direct solution algorithm. Currently, the best lower bound on the computational complexity of determining whether Player~$O$ wins a delay game with max-regular winning condition for some constant delay function is the $\exptime$-hardness result for games with safety conditions~\cite{KleinZimmermann15}. 

Also, we showed that constant delay is not sufficient for max-regular conditions by giving a condition~$L$ such that Player~$O$ wins $\delaygame{L}$ for some linear $f$, but not for any constant delay function~$f$. 

Both the lower bound on the necessary lookahead and the one on the computational complexity for safety conditions mentioned above are complemented by matching upper bounds for games with parity conditions~\cite{KleinZimmermann15}, i.e., having a parity condition instead of a safety condition has no discernible influence. Stated differently, the complexity of the problems manifests itself in the transition structure of the automaton. Our example from Section~\ref{sec_linear} shows that this is no longer true for max-regular conditions: having a quantitative acceptance condition requires growing lookahead.

In ongoing work, we aim to solve delay game with respect to arbitrary delay functions and to determine whether linear delay functions are sufficient to win delay games with max-regular winning conditions.


\bibliographystyle{splncs03}
\bibliography{biblio}


\end{document}